\documentclass[12pt,onecolumn,draftcls]{IEEEtran} 
\usepackage{amsfonts,amssymb,amsmath,amsthm}
\usepackage{algorithmic,algorithm}
\usepackage{xcolor}
\usepackage{ifpdf}
\ifpdf
\usepackage[]{graphicx}
\else
\usepackage[]{graphicx}
\fi
\usepackage{url,cite}
\usepackage{changepage}
\newtheorem{algorithms}{Algorithm}
\newtheorem{corollary}{Corollary}

\newtheorem{proposition}{Proposition}
\newtheorem{remark}{Remark}
\newtheorem{theorem}{Theorem}
\def\det{\mbox{\scriptsize det}}

\def\aa{\hspace{5mm}}
\def\CH{\mathsf{CH}}
\def\ch{\mathsf{ch}}
\def\CL{\mathsf{CL}}
\def\cl{\mathsf{cl}}

\newenvironment{namelist}[1]{%
\begin{list}{}
  {
   
   \settowidth{\labelwidth}{#1}
   \setlength{\leftmargin}{1.1\labelwidth}
   }
  }{%
\end{list}}
\newcommand{\step}[1]{\mbox{\rm\bf Step {#1}:}}
\usepackage[normalem]{ulem}

\begin{document}
\title{Optimal Construction of Regenerating Code through Rate-matching in Hostile Networks}
\author{Jian Li\aa Tongtong Li\aa Jian Ren\thanks{The authors are with the Department of ECE, Michigan State University, East Lansing, MI 48824-1226. 
 Email: \{lijian6, tongli, renjian\}@msu.edu}}

\date{November 4, 2015}
\maketitle

\begin{abstract}
Regenerating code is a class of code very suitable for distributed storage systems, which can maintain optimal bandwidth and storage space. Two types of important regenerating code have been constructed: the minimum storage regeneration (MSR) code and the minimum bandwidth regeneration (MBR) code. However, in hostile networks where adversaries can compromise storage nodes, the storage capacity of the network can be significantly affected. In this paper, we propose two optimal constructions of regenerating codes through rate-matching that can combat against this kind of adversaries in hostile networks: 2-layer rate-matched regenerating code and $m$-layer rate-matched regenerating code. For the 2-layer code, we can achieve the optimal storage efficiency for given system requirements. Our comprehensive analysis shows that our code can detect and correct malicious nodes with higher storage efficiency compared to the universally resilient regenerating code which is a straightforward extension of regenerating code with error detection and correction capability. Then we propose the $m$-layer code by extending the 2-layer code and achieve the optimal error correction efficiency by matching the code rate of each layer's regenerating code. We also demonstrate that the optimized parameter can achieve the maximum storage capacity under the same constraint. Compared to the universally resilient regenerating code, our code can achieve much higher error correction efficiency.
\end{abstract}


\section{Introduction}
Distributed storage is a popular method to store files securely without requiring data encryption. Instead of storing a file and its replications in multiple servers, we can break the file into components and store the components into multiple servers. 
In this way, both the reliability and the security of the file can be increased. A typical approach is to encode the file using an $(n,k)$ Reed-Solomon (RS) code and distribute the encoded file into $n$ servers. When we need to recover the file, we only need to collect the encoded parts from $k$ servers, which achieves a trade-off between reliability and efficiency. However, when repairing or regenerating the contents of a failed node, the whole file has to be recovered first, which is a waste of bandwidth. 

The concept of regenerating code was introduced in~\cite{Dimakis}, 
where a replacement node is allowed to connect to some individual nodes directly and regenerate a substitute of the failed node, instead of first recovering the original data then regenerating the failed component. Compared to the RS code, regenerating code achieves an optimal tradeoff between bandwidth and storage within the minimum storage regeneration (MSR) and the minimum bandwidth regeneration (MBR) points.  

However, when malicious behaviors exist in the network, both the regeneration of the failed node or the reconstruction of the original file will fail. The error resilience of the Reed-Solomen code based regenerating code in the network with errors and erasures was analyzed in~\cite{Rashmi-err}. In our previous work, a Hermitian code based regenerating code was proposed to provide better error correction capability compared to the Reed-Solomen code based approach.



Inspired by the nice performance of Hermitian code based regenerating codes, in this paper we step forward to further construct optimal regenerating codes which have similar layered structure like Hermitian code in distributed storage. The main contributions of this paper are:
\begin{itemize}
\item We propose an optimal construction of 2-layer rate-matched regenerating code. Both theoretical analysis and performance evaluation show that this code can achieve storage efficiency higher than the universally resilient regenerating code proposed in~\cite{Rashmi-err}.

\item We propose an optimal construction of $m$-layer rate-matched regenerating code. The $m$-layer code can achieve higher error correction efficiency than the code proposed in~\cite{Rashmi-err} and the Hermitian code based regenerating code proposed in~\cite{jian14}. Furthermore, the $m$-layered code is  easier to understand and has more flexibility than the Hermitian based code.
\end{itemize}
Here we will focus on error correction and malicious node locating in data regeneration and reconstruction in distributed storage. When no error occurs or no malicious node exists, the data regeneration and reconstruction can be processed the same as the existing works. 

It it worth to note that although there are two types of regenerating codes: MSR code and MBR code on the MSR point and MBR point respectively, in this paper we will only focus on the optimization of the MSR code for the following two reasons:
\begin{enumerate}
\item The processes and results of the optimization for these two codes are similar. The optimization for the MSR code can be directly applied to the MBR code with similar optimization results.

\item The differences between the constructions of MSR code and MBR code have little impact on the optimization proposed in this paper.
\end{enumerate}

The rest of this paper is organized as follows: in Section~\ref{Sec:related} we introduce the related work. In Section~\ref{Sec:Preliminary}, the preliminary of this paper is presented.  In Section~\ref{Sec:rate-matched-MSR}, we propose two component codes for the rate-matched regenerating codes. We propose and analyze the 2-layer rate-matched regenerating code in Section~\ref{sec:2-layer}. Then we propose and analyze the $m$-layer rate-matched regenerating code in Section~\ref{sec:m_layer_msr}. The paper is concluded in Section~\ref{Sec:Conclusion}.

\section{Related Work}\label{Sec:related}

When a storage node in the distributed storage network that employing the conventional $(n,k)$ RS code (such as OceanStore~\cite{Ocean} and Total Recall~\cite{Total}) fails,
the replacement node connects to $k$ nodes and downloads the whole file to recover the symbols stored in the failed node. This approach is a waste of bandwidth because the whole file has to be downloaded to recover a fraction of it.
To overcome this drawback, Dimakis \emph{et al}.~\cite{Dimakis} introduced the conception of $\{n, k, d, \alpha, \beta, B\}$ regenerating code based on the network coding. In the context of regenerating code, the contents stored in a failed node can be regenerated by the replacement node through downloading $\gamma$ help symbols from $d$ helper nodes. The bandwidth consumption for the failed node regeneration could be far less than the whole file. A data collector~(DC)
can reconstruct the original file stored in the network by downloading $\alpha$ symbols from each of the $k$ storage nodes.
In~\cite{Dimakis}, the authors proved that there is a tradeoff between bandwidth $\gamma$ and per node storage $\alpha$.  They found two optimal points: minimum storage regeneration (MSR) and minimum bandwidth regeneration (MBR) points. Currently there are many literatures focusing on the optimal regenerating codes design: \cite{Daniel09searchingformsr,Shah10explicitmbr,Changho10exactmdsia,Yunnan11constructionsysmdsmbr,Papailiopoulos12simrc,El10fracrepetitioninds,Tamo11mdsarraycodes,Viveck11optimalrepairviaia,Papailiopoulos13repairhadamard,Shah10flexiblercfords,Shum11existmbrcorc,Anyu13exactcorcmbr}. In~\cite{Hou13basicrcbinary,Yuliang09rcbasedp2p} the implementation of the regenerating code were studied.

The regenerating code can be divided into functional regeneration and exact regeneration. In the functional regeneration, the replacement node
regenerates a new component that can functionally replace the failed component instead of being the same as the original stored component. \cite{Ywu} formulated the data regeneration as a multicast network coding problem and constructed functional regenerating codes. \cite{Duminuco} implemented a random linear regenerating codes for distributed storage systems. \cite{Shum} proved that by allowing data exchange among the replacement nodes, a better tradeoff between repair bandwidth $\gamma$ and per node storage $\alpha$ can be achieved. In the exact regeneration, the replacement node regenerates the exact symbols of a failed node. \cite{Ywu2} proposed to reduce the regeneration bandwidth through algebraic alignment. \cite{Shah} provided a code structure for exact regeneration using interference alignment technique. \cite{Rashmi} presented optimal exact constructions of MBR codes and MSR codes under product-matrix framework. This is the first work that allows independent selection of the nodes number $n$ in the network.

None of these works above considered code regeneration under node corruption or adversarial manipulation attacks in hostile networks. In fact, all these schemes will fail in both regeneration and reconstruction if there are nodes in the storage cloud sending out incorrect responses to the regeneration and reconstruction requests.

In~\cite{Oggier11byzanfaulttolofrc}, the Byzantine fault tolerance of regenerating codes were studied.
In~\cite{Pawar}, the authors discussed the amount of information that can be safely stored against passive eavesdropping and active adversarial attacks based on the regeneration structure.
In~\cite{Han}, the authors proposed to add CRC codes in the regenerating code to check the integrity of the data in hostile networks. Unfortunately, the CRC checks can also be manipulated by the malicious nodes, resulting in the failure of the regeneration and reconstruction.
In~\cite{Chen12dataintegrityprotection}, the authors proposed to add data integrity protection in distributed storage.
In~\cite{Cachin06optimalresforerasurecoded}, the authors proposed an erasure-coded distributed storage based on threshold cryptography.
In~\cite{Abd05lazyver}, the authors analyzed the verification cost for both the client read and write operation in workloads with idle periods.
In~\cite{Rashmi-err}, the authors analyzed the error resilience of the RS code based regenerating code in the network with errors and erasures. They provided the theoretical error correction capability.  
 In \cite{jian14} the authors proposed a Hermitian code based regenerating code, which could provide better error correction capability. In~\cite{Rashimi-universal} the authors proposed the universally secure regenerating code to achieve information theoretic data confidentiality. But the extra computational cost and bandwidth have to be considered for this code. In \cite{jian14glo} the authors proposed to apply linear feedback shift register (LFSR) to protect the data confidentiality.

\section{Preliminary and Assumptions} \label{Sec:Preliminary}
\label{Sec:Preliminary}

\subsection{Regenerating Code}
Regenerating code introduced in~\cite{Dimakis} is a linear code over finite filed $\mathbb{F}_q$ with a set of parameters $\{n, k, d, \alpha, \beta, B\}$. A file of size $B$ is stored in $n$ storage nodes, each of which stores $\alpha$ symbols. A replacement node can regenerate the contents of a failed node by downloading $\beta$ symbols from each of $d$ randomly selected storage nodes. So the total bandwidth needed to regenerate a failed node is $\gamma = d\beta$. The data collector (DC) can reconstruct the whole file by downloading $\alpha$ symbols from each of $k\leq d$ randomly selected storage nodes. In~\cite{Dimakis}, the following theoretical bound was derived:
\begin{equation}
\label{eq:min_cut}
B \leq \sum_{i=0}^{k-1}\min \{ \alpha, (d-i)\beta \}.
\end{equation}
From equation~(\ref{eq:min_cut}), a trade-off between the regeneration bandwidth $\gamma$ and the storage requirement $\alpha$ was derived. $\gamma$ and $\alpha$ cannot be decreased at the same time. There are two special cases: minimum storage regeneration (MSR) point in which the storage parameter $\alpha$ is minimized;
\begin{equation}
\label{eq:MSR_tradeoff}
(\alpha_{MSR},\gamma_{MSR})= \left(\frac Bk, \frac{Bd}{k(d-k+1)}\right),
\end{equation}
and minimum bandwidth regeneration (MBR) point in which the bandwidth $\gamma$ is minimized:
\begin{equation}
\label{eq:MBR_tradeoff}
(\alpha_{MBR},\gamma_{MBR})= \left(\frac{2Bd}{2kd-k^2 + k},\frac{2Bd}{2kd-k^2 + k} \right).
\end{equation}

%
\subsection{System Assumptions and Adversarial Model}
In this paper, we assume there is a secure server that is responsible for encoding and distributing the data to storage nodes. Replacement nodes will also be initialized by the secure server. DC and the secure server can be implemented in the same computer and can never be compromised. We use the notation $\CH$/$\CL$ to refer to either the full rate/fractional rate MSR code or a codeword of the full rate/fractional rate MSR code. The exact meaning can be discriminated clearly according to the context.

We assume some network nodes may be corrupted due to hardware failure or communication errors, and/or be compromised by malicious users.  As a result, upon request, these nodes may send out incorrect responses to disrupt the data regeneration and reconstruction. 
The adversary model is the same as~\cite{Rashmi-err}, 
We assume that the malicious users can take full control of $\tau$ ($\tau \leq n$ and corresponds to $s$ in~\cite{Rashmi-err})  storage nodes and collude to perform attacks. 

We will refer these symbols as \emph{bogus} symbols without making distinction between the corrupted symbols and compromised symbols. We will also use corrupted nodes, malicious nodes and compromised nodes interchangeably without making any distinction.

\section{Component Codes of Rate-matched Regenerating Code}\label{Sec:rate-matched-MSR}
In this section, we will introduce two different component codes for rate-matched MSR code on the MSR point with $d = 2k-2$. The code based on the MSR point with $d > 2k -2$ can be derived the same way through truncating operations. In the rate-matched MSR code, there are two types of MSR codes with different code rates: full rate code and fractional rate code. 

\subsection{Full Rate Code}

\subsubsection{Encoding}

The full rate code is encoded based on the product-matrix code framework in~\cite{Rashmi}. According to equation~(\ref{eq:MSR_tradeoff}), we have $\alpha_H = d/2$, $\beta_H = 1$ for one block of data with the size $B_H=(\alpha+1)\alpha$. The data will be arranged into two $\alpha \times \alpha$ symmetric matrices $S_1,S_2$, each of which will contain $B_H/2$ data. The codeword $\CH$ is defined as
\begin{equation}
\label{eq:encoding_msr_h}
\CH = [\Phi \: \:\: \Lambda\Phi]
\begin{bmatrix}
S_1 \\
S_2
\end{bmatrix}
= \Psi M_H=\begin{bmatrix}\ch_1\\ \vdots\\ \ch_n\end{bmatrix},
\end{equation}
where
\begin{equation}\label{eq:phi}
\Phi =
\begin{bmatrix}
1 & 1 & 1 & \dots  & 1 \\
1 & g & g^2 & \dots  & g^{\alpha-1} \\
\vdots & \vdots & \vdots & \ddots  & \vdots \\
1 & g^{n-1}& (g^{n-1})^2 & \dots  & (g^{n-1})^{\alpha-1} 
\end{bmatrix}
\end{equation}
is a Vandermonde matrix and $\Lambda=\mbox{diag}[\lambda_1,\lambda_2,\cdots,\lambda_n]$ such that $\lambda_i\in \mathbb{F}_q$ and $\lambda_i\ne\lambda_j$ for $1\leq i, j\leq n, i\ne j$, $g$ is a primitive element in $\mathbb{F}_q$, and any $d$ rows of $\Psi$ are linearly independent. Then each row $\ch_i=\boldsymbol{\psi}_iM_H$ ($0 \leq i < n$) of the codeword matrix $\CH$ will be stored in storage node $i$, where the encoding vector $\boldsymbol{\psi}_i$ is the $i^{th}$ row of $\Psi$. 

\subsubsection{Regeneration} \label{Sec:Reg_high_rate}

Suppose node $z$ fails, the replacement node $z'$ will send regeneration requests to the rest of $n-1$ helper nodes. Upon receiving the regeneration request, helper node $i$ will calculate and send out the help symbol $p_i = \ch_i \boldsymbol{\phi}_z^T=\boldsymbol{\psi}_iM_H\boldsymbol{\phi}_z^T$, where $\boldsymbol{\phi}_z$ is the $z^{th}$ row of $\Phi$. $z'$ will perform Algorithm~\ref{alg:reg_h} to regenerate the contents of the failed node $z$.  For convenience, we define
$\Psi_{i\to j}=\begin{bmatrix}
\boldsymbol{\psi}_i^T,
\boldsymbol{\psi}_{i+1}^T
\cdots,
\boldsymbol{\psi}_j^T
\end{bmatrix}^T,
$
where $\boldsymbol{\psi}_t$ is the $t^{th}$ row of $\Psi$ $(i\leq t\leq j)$ and $\mathbf{x}^{(j)}$ is the vector containing the first $j$ symbols of $M_H \boldsymbol{\phi}_z^T$. 

Suppose $p'_i = p_i + e_i$ is the response from the $i^{th}$ helper node. If $p_i$ has been modified by the malicious node $i$, we have $e_i \in{\mathbb{F}_q}\backslash \{0\}$. We can successfully regenerate the symbols in node $z$ when the number of errors in the received help symbols ${p_i}'$ from $n-1$ helper nodes is less than $\lfloor (n-d-1)/2\rfloor$, where $\left\lfloor \cdot\right\rfloor$ is the floor operation. Without loss of generality, we assume $0 \leq i \leq n-2$.

\begin{algorithms}$z'$ regenerates symbols of the failed node $z$ 
\label{alg:reg_h}
\normalfont
\begin{namelist}{\textbf{Step n:}}
\item [\step{1}] Decode $\mathbf{p}'$ to $\mathbf{p}_{cw}$, where $\mathbf{{p}}'=[p'_0, p'_1, \cdots,  p'_{n-2}]^T$ can be viewed as an MDS code with parameters $(n-1,d,n-d)$ since $\Psi_{0\to (n-2)} \cdot \mathbf{{x}}^{(n-1)}  = \mathbf{p}'$.

\item [\step{2}] Solve $\Psi_{0\to (n-2)} \cdot \mathbf{{x}}^{(n-1)}  = \mathbf{p}_{cw}$ and compute $\ch_z=\boldsymbol{\phi}_z S_1 + \lambda_z  \boldsymbol{\phi}_z S_2$ as described in~\cite{Rashmi}.
\end{namelist}
\end{algorithms}

\begin{proposition}
For regeneration, the full rate code can correct errors from $\left\lfloor (n-d-1)/2 \right\rfloor$ malicious nodes, where $\left\lfloor \cdot\right\rfloor$ is the floor operation.
\end{proposition}

\subsubsection{Reconstruction}\label{Sec:Rec_high_rate}

When the DC needs to reconstruct the original file, it will send reconstruction requests to $n$ storage nodes. Upon receiving the request, node $i$ will send out the symbol vector $\mathbf{c}_i$ to the DC. Suppose $\mathbf{c}'_i= \mathbf{c}_i + \mathbf{e}_i$ is the response from the $i^{th}$ storage node. If $\mathbf{c}_i$ has been modified by the malicious node $i$, we have $\mathbf{e}_i \in \mathbb{F}_q^{\alpha} \backslash \{\mathbf{0}\}$. 

The DC will reconstruct the file as follows: Let
$R' = [{\ch'_{0}}^T, {\ch'_{1}}^T, \cdots, {\ch'_{n-1}}^T]^T$, we have
\begin{equation*}
\Psi 
\begin{bmatrix}
S'_1\\
S'_2
\end{bmatrix} 
=[\Phi \: \:\: \Lambda\Phi] 
\begin{bmatrix}
S'_1\\
S'_2
\end{bmatrix} 
=R',
\end{equation*}
\begin{equation}\label{eq:rec_eqn}
\Phi  S'_1  \Phi^T + \Lambda  \Phi  S'_2  \Phi^T =   {R}'\Phi^T.
\end{equation}

Let $C=\Phi  S'_1  \Phi^T$, $D=\Phi  S'_2  \Phi^T$, and ${\widehat{R}}'={R}'\Phi^T$,  then
\begin{equation}
C + \Lambda  D  =   {\widehat{R}}'.
\end{equation}
Since $C,D$ are both symmetric, we can solve the non-diagonal elements of $C,D$ as follows:
\begin{equation}
\label{eq:recon_recovery}
\left\{\begin{matrix}
C_{i,j} + \lambda_i \cdot D_{i,j}= {\widehat{R}}'_{i,j}\\
C_{i,j} + \lambda_j \cdot D_{i,j}= {\widehat{R}}'_{j,i}.
\end{matrix}\right.
\end{equation}
Because matrices $C$ and $D$ have the same structure, here we only focus on $C$ (corresponding to $S'_1$). It is straightforward to see that if node $i$ is malicious and there are errors in the $i^{th}$ row of $R'$, there will be errors in the $i^{th}$ row of ${\widehat{R}}'$. Furthermore, there will be errors in the $i^{th}$ row and $i^{th}$ column of $C$. Define $S'_1 \Phi^T=\widehat{S}'_1$, we have $\Phi \widehat{S}'_1 = C$.
Here we can view each column of $C$ as an $(n-1, \alpha, n - \alpha)$ MDS code because $\Phi$ is a Vandermonde matrix. The length of the code is $n - 1$ since the diagonal elements of $C$ is unknown.
Suppose node $j$ is a legitimate node, we can decode the MDS code to recover the $j^{th}$ column of $C$ and locate the malicious nodes. Eventually $C$ can be recovered. 
So the DC can reconstructs $S_1$ using the method similar to~\cite{Rashmi,jian14},  For $S_2$, the recovering process is similar.

\begin{proposition}
For reconstruction, the full rate code can correct errors from $\left\lfloor (n-\alpha-1) /2 \right\rfloor$ malicious nodes.
\end{proposition}

\subsection{Fractional Rate Code}

\subsubsection{Encoding}
For the fractional rate code, we also have $\alpha_L = d/2$, $\beta_L = 1$ for one block of data with the size 
\setlength\arraycolsep{0pt}%
\begin{equation}\label{eq:BL}
B_L\!=\!\left\{
\begin{matrix}
&xd(1+xd)/2, &x \in (0,0.5]\\ 
&\!\alpha(\alpha \! \!+ \!\!1)/2\! +\!\! (x \!-\!0.5)d(1\!\!+\!\!(x\!-\!0.5)d)/2, &x \in\!\! (0.5,\!1] \end{matrix}\right.,
\end{equation}
where $x$ is the match factor of the rate-matched MSR code. It is easy to see that the fractional rate code will become the full rate code with $x=1$. The data ${\mathbf{m}} = [m_1,m_2,\dots,m_{B_L}] \in(\mathbb{F}_q)^{B_L}$ will be processed as follows: 

When $x \leq 0.5$, the data will be arranged into a symmetric matrix $S_1$ of the size $\alpha \times \alpha$:
\setlength\arraycolsep{3pt}%
\begin{equation}
S_1 = \begin{bmatrix}
m_1 & m_2 & \dots &  m_{xd} &0 &\dots &0\\
m_2 & m_{xd+1} & \dots  & m_{2xd-1} &0 &\dots &0\\
\vdots & \vdots  & \ddots  & \vdots &\vdots  & \ddots  & \vdots\\
m_{xd} & m_{2xd-1} & \dots  &  m_{B_L/2}&0 &\dots &0\\
0 & 0 &\dots & 0 & 0 &\dots & 0\\
\vdots & \vdots  & \ddots  & \vdots &\vdots  & \ddots  & \vdots\\
0 & 0 &\dots & 0 & 0 &\dots & 0
\end{bmatrix}.
\end{equation}
 The codeword $\CL$ is defined as
\begin{equation}
\label{eq:encoding_msr_l}
\CL = [\Phi \: \:\: \Lambda\Phi]
\begin{bmatrix}
S_1 \\
\mathbf{0}
\end{bmatrix}
= \Psi M_L,
\end{equation}
where $\mathbf{0}$ is the $\alpha \times \alpha$ zero matrix and $\Phi, \Lambda, \Psi$ are the same as the full rate code. 

When $x>0.5$, the first $\alpha(\alpha+1)/2$ data will be arranged into an $\alpha \times \alpha$ symmetric matrix $S_1$. The rest of the data $m_{\alpha(\alpha+1)/2 +1}, \dots, m_{B_L}$ will be arranged into another  $\alpha \times \alpha$ symmetric matrix $S_2$:
\setlength\arraycolsep{1.9pt}%
\begin{equation}
S_2 = \begin{bmatrix}
m_{\alpha(\alpha+1)/2 +1} &  \dots &  m_{\alpha(\alpha+1)/2 +xd} &0 &\dots &0\\
m_{\alpha(\alpha+1)/2 +2} &  \dots  & m_{\alpha(\alpha+1)/2 +2xd-1} &0 &\dots &0\\
\vdots &  \ddots  & \vdots &\vdots  & \ddots  & \vdots\\
m_{\alpha(\alpha+1)/2 +xd} &  \dots  &  m_{B_L/2}&0 &\dots &0\\
0 &\dots & 0 & 0 &\dots & 0\\
\vdots   & \ddots  & \vdots &\vdots  & \ddots  & \vdots\\
0 & \dots & 0 & 0 &\dots & 0
\end{bmatrix}.
\end{equation}
The codeword $\CL$ is defined the same as equation~(\ref{eq:encoding_msr_h}) with the same parameters  $\Phi, \Lambda$ and $\Psi$.

Then each row $\cl_i$ ($0 \leq i < n$) of the codeword matrix $\CL$ will be stored in storage node $i$ respectively, in which the encoding vector $\boldsymbol{\psi}_i$ is the $i^{th}$ row of $\Psi$.

\begin{proposition}
The fractional rate code can achieve the MSR point in equation~(\ref{eq:MSR_tradeoff}) since it it encoded under the product-matrix MSR code framework in~\cite{Rashmi}.
\end{proposition}

\subsubsection{Regeneration}\label{sec:reg:frac}

The regeneration for the fractional rate code is the same as the regeneration for the full rate code described in Section~\ref{Sec:Reg_high_rate} with only a minor difference. If we define $\mathbf{x}^{(j)}$ as the vector containing the first $j$ symbols of $M_L \boldsymbol{\phi}_z^T$, there will be only $xd$ nonzero elements in the vector. According to $\boldsymbol{\Psi}_{0\to {n-2}} \cdot \mathbf{{x}}^{(n-1)}  = \mathbf{p}'$, the received symbol vector $\mathbf{p}'$ for the fractional rate code in {\bf{Step 1}} of Algorithm~\ref{alg:reg_h} can be viewed as an $(n-1,xd,n-xd)$ MDS code. Since $x<1$, we can detect and correct more errors in data regeneration using the fractional rate code than using the full rate code. 
\begin{proposition}
For regeneration, the fractional rate code can correct errors from $\left\lfloor (n-xd-1) /2 \right\rfloor$ malicious nodes.
\end{proposition}

\subsubsection{Reconstruction}

The reconstruction for the fractional rate code is similar to that for the full rate code described in Section~\ref{Sec:Rec_high_rate}. Let
$R' = [{\cl'}_0^T, {\cl'}_1^T, \cdots, {\cl'}_{n-1}^T]^T$. 

When the match factor $x>0.5$, reconstruction for the fractional rate code is the same to that for the full rate code.

When $x \le 0.5$, equation~(\ref{eq:rec_eqn}) can be written as:
\begin{equation}\label{eq:rec_eqn_low}
\Phi  S'_1 =   R'.
\end{equation}
So we can view each column of $R'$ as an $(n,xd,n-xd+1)$ MDS code. After decoding $R'$ to $R_{cw}$, we can recover the data matrix $S_1$ by solving the equation $\Phi  S_1 =   R_{cw}.$ Meanwhile, if the $i^{th}$ rows of $R'$ and $R_{cw}$ are different, we can mark node $i$ as corrupted.
\begin{proposition}
For reconstruction, when the match factor $x>0.5$, the fractional rate code can correct errors from $\left\lfloor (n-\alpha-1) /2 \right\rfloor$ malicious nodes. When the match factor $x \leq 0.5$, the fractional rate code can correct errors from $\left\lfloor (n-xd) /2 \right\rfloor$ malicious nodes.
\end{proposition}

\section{2-Layer Rate-matched regenerating Code}\label{sec:2-layer}

In this section, we will show our first optimization of the rate-matched MSR code: 2-layer rate-matched MSR code. In the code design, we utilize two layers of the MSR code: the fractional rate code for one layer and the full rate code for the other. The purpose of the fractional rate code is to correct the erroneous symbols sent by malicious nodes and locate the corresponding malicious nodes. Then we can treat the errors in the received symbols as erasures when regenerating with the full rate code. However, the rates of the two codes must match to achieve an optimal performance. Here we mainly focus on the rate-matching for data regeneration. We can see in the later analysis that the performance of data reconstruction can also be improved with this design criterion. 


We will first fix the error correction capabilities of the full rate code and the fractional rate code.  Then we will derive the optimal rate matching criteria to optimize the data storage efficiency under the fixed error correction capability.

\subsection{Rate Matching}

From the analysis above, we know that during regeneration, the fractional rate code can correct up to $\left \lfloor (n-xd-1)/2  \right \rfloor$ errors, which are more than $\left \lfloor (n-d-1)/2  \right \rfloor$ errors that the full rate code can correct. In the 2-layer rate-matched MSR code design, our goal is to match the fractional rate code with the full rate code. The main task for the fractional rate code is to detect and correct errors, while the main task for the full rate code is to maintain the storage efficiency. So if the fractional rate code can locate all the malicious nodes, the full rate code can simply treat the symbols received from these malicious nodes as erasures, which requires the minimum redundancy for the full rate code. The full rate code can correct up to $n-d-1$ erasures. Thus we have the following optimal rate-matching equation:
\begin{equation}\label{eq:rate-match}
\left \lfloor (n-xd-1)/2  \right \rfloor = n-d-1,
\end{equation}
from which we can derive the match factor $x$.

\subsection{Encoding}

To encode a file with size $B_F$ using the 2-layer rate-matched MSR code, the file will first be divided into $\theta_H$ blocks of data with the size $B_H$ and $\theta_L$ blocks of data with the size $B_L$, where the parameters should satisfy
\begin{equation}
\label{eq:size_equation}
B_F = \theta_H B_H +\theta_L B_L.
\end{equation}
Then the $\theta_H$ blocks of data will be encoded into code matrices $\CH_1,\dots,\CH_{\theta_H}$ using the full rate code and the $\theta_L$ blocks of data will be encoded into code matrices $\CL_1,\dots,\CL_{\theta_L}$ using the fractional rate code. To prevent the malicious nodes from corrupting the fractional rate code only, the secure server will randomly concatenate all the matrices together to form the final $n \times \alpha(\theta_H + \theta_L)$ codeword matrix:
\begin{equation}
\label{eq:final_codeword_matrix}
\mathsf{CM} = [\mbox{Perm}(\CH_1,\dots,\CH_{\theta_H},\CL_1,\dots,\CL_{\theta_L})],
\end{equation}
where $\mbox{Perm}(\cdot)$ is the random matrices permutation operation. The secure sever will also record the order of the permutation for future code regeneration and reconstruction. Then each row $\mathbf{c}_i = [\mbox{Perm}(\ch_{1,i}, \dots, \ch_{\theta_H,i},\cl_{1,i}, \dots, \cl_{\theta_L,i})]$ ($0 \leq i \leq n-1$) of the codeword matrix $\mathsf{CM}$ will be stored in storage node $i$, where $\ch_{j,i}$ is the $i^{th}$ row of $\CH_j$ ($1 \leq j \leq \theta_H$), and  $\cl_{j,i}$ is the $i^{th}$ row of $\CL_j$ ($1 \leq j \leq \theta_L$). The encoding vector $\boldsymbol{\psi}_i$ for storage node $i$ is the $i^{th}$ row of $\Psi$ in equation~(\ref{eq:encoding_msr_h}). Therefore, we have the following Theorem.
\begin{theorem}
The encoding of 2-layer rate-matched MSR code can achieve the MSR point in equation~(\ref{eq:MSR_tradeoff}) since both the full rate code and the fractional code are MSR codes.
\end{theorem}

\subsection{Regeneration}

Suppose node $z$ fails, the security server will initialize a replacement node $z'$ with the order information of the fractional rate code and the full rate code in the 2-layer rate-matched MSR code. Then the replacement node $z'$ will send regeneration requests to the rest of $n-1$ helper nodes. Upon receiving the regeneration request, helper node $i$ will calculate and send out the help symbol $p_i = {\bf{c}}_i \boldsymbol{\phi}_z^T$. $z'$ will perform Algorithm~\ref{alg:reg_rate_matched} to regenerate the contents of the failed node $z$. After the regeneration is finished, $z'$ will erase the order information. So even if $z'$ was compromised later, the adversary would not get the permutation order of the fractional rate code and the full rate code.

\begin{algorithms}$z'$ regenerates symbols of the failed node $z$ for the 2-layer rate-matched MSR code
\label{alg:reg_rate_matched}
\normalfont
\begin{namelist}{\textbf{Step n:}}
\item [\step{1}] According to the order information, regenerate all the symbols related to the $\theta_L$ data blocks encoded by the fractional rate code, using Algorithm~\ref{alg:reg_h}. If errors are detected in the symbols sent by node $i$, it will be marked as a malicious node.

\item [\step{2}] Regenerate all the symbols related to the $\theta_H$ data blocks encoded by the full rate code, using Algorithm~\ref{alg:reg_h}. During the regeneration, all the symbols sent from nodes marked as malicious nodes will be replaced by erasures $\bigotimes$.
\end{namelist}
\end{algorithms}

It is easy to see that Algorithm~\ref{alg:reg_rate_matched} can correct errors and locate malicious node using the fractional rate code while achieve high storage efficiency using the full rate code. We summarize the result as the following Theorem.
\begin{theorem}
For regeneration, the 2-layer rate-matched MSR code can correct errors from $\left\lfloor (n-xd-1) /2 \right\rfloor$ malicious nodes.
\end{theorem}

\subsection{Parameters Optimization}\label{sec:para_opt_reg}

We have the following design requirements for a given distributed storage system applying the 2-layer rate-matched MSR code:
\begin{itemize}
\item The maximum number of malicious nodes $M$ that the system can detect and locate using the fractional rate code. We have 
\begin{equation}\label{eq:error_correction}
\left \lfloor  (n-xd-1)/2\right \rfloor = M.
\end{equation}
\item The probability $P_{\det}$ that the system can detect all the malicious nodes. The detection will be successful if each malicious node modifies at least one help symbol corresponding to the fractional rate code and sends it to the replacement node. Suppose the malicious nodes modify each help symbol to be sent to the replacement node with probability $P$, we have
\begin{equation}\label{eq:det_suc}
 (1-(1-P)^{\theta_L})^M \ge P_{\det}.
\end{equation}
\end{itemize}
So there is a trade-off between $\theta_L$ and $\theta_H$: the number of data blocks encoded by the fractional rate code and the number of data blocks encoded by the full rate code. If we encode using too much full rate code, we may not meet the detection probability $P_{\det}$ requirement. If too much fractional rate code is used, the redundancy may be too high.

The storage efficiency is defined as the ratio between the actual size of data to be stored and the total storage space needed by the encoded data:
\begin{equation}\label{eq:efficiency}
\delta_S = \frac{\theta_H B_H + \theta_L B_L}{(\theta_H + \theta_L)n\alpha} = \frac{B_F}{(\theta_H + \theta_L)n\alpha}.
\end{equation}
Thus we can calculate the optimized parameters $x$, $d$, $\theta_H$, $\theta_L$  by maximizing equation~(\ref{eq:efficiency}) under the constraints defined by equations~(\ref{eq:rate-match}), (\ref{eq:size_equation}), (\ref{eq:error_correction}), (\ref{eq:det_suc}).

$d$ and $x$ can be determined by equation (\ref{eq:rate-match}) and (\ref{eq:error_correction}):
\begin{eqnarray}
d &=& n - M - 1,\\
x &=& (n - 2M -1) / (n - M -1).
\end{eqnarray}
Since $B_F$ is constant, to maximize $\delta_S$ is equal to minimize $\theta_H + \theta_L$. So we can rewrite the optimization problem as follows:
\begin{equation}
\mbox{Minimize} \:\:\: \theta_H + \theta_L,\:\:\: \mbox{subject to}~(\ref{eq:size_equation})\:\: \mbox{and} \:\:(\ref{eq:det_suc}).
\end{equation}
This is a simple linear programming problem. Here we will show the optimization results directly:
\begin{eqnarray}
\theta_L &=& \log_{(1-P)}(1-P_{\det}^{1/M}),\\
\theta_H &=& (B_F -  \theta_L B_L) / B_H.
\end{eqnarray}
In this paper we assume that we are storing large files, which means $B_F >  \theta_L B_L $. So an optimal solution for the 2-layer rate-matched MSR code can always be found. We have the following theorem:
\begin{theorem}
When the number of blocks of the fractional rate code $\theta_L$ equals to $\log_{(1-P)}(1-P_{\det}^{1/M})$ and the number of blocks of the full rate code $\theta_H$ equals to $(B_F -  \theta_L B_L) / B_H$, the 2-layer rate-matched MSR code can achieve the optimal storage efficiency.
\end{theorem}

\subsection{Reconstruction}
When DC needs to reconstruct the original file, it will send reconstruction requests to $n$ storage nodes. Upon receiving the request, node $i$ will send out the symbol vector $\mathbf{c}_i$. Suppose $\mathbf{c}'_i= \mathbf{c}_i + \mathbf{e}_i$ is the response from the $i^{th}$ storage node. If $\mathbf{c}_i$ has been modified by the malicious node $i$, we have $\mathbf{e}_i \in \mathbb{F}_q^{\alpha(\theta_L + \theta_H)} \backslash \{\mathbf{0}\}$. Since DC has the permutation information of the fractional rate code and the full rate code, similar to the regeneration of the 2-layer rate-matched MSR code, DC will perform the reconstruction using Algorithm~\ref{alg:rec_rate_matched}.

\begin{algorithms}DC reconstructs the original file for the 2-layer rate-matched MSR code
\label{alg:rec_rate_matched}

\normalfont

\begin{namelist}{\textbf{Step n:}}
\item [\step{1}] According to the order information, reconstruct each of the $\theta_L$ data blocks encoded by the fractional rate code and locate the malicious nodes. 

\item [\step{2}] Reconstruct each of the data blocks encoded by the full rate code. During the reconstruction, all the symbols sent from malicious nodes will be replaced by erasures $\bigotimes$.
\end{namelist}
\end{algorithms}

In Section~\ref{sec:para_opt_reg}, we optimized the parameters for the data regeneration, considering the trade-off between the successful malicious node detection probability and the storage efficiency. For data reconstruction, we have the following theorem:
\begin{theorem}[Optimized Parameters]
When the number of blocks of the fractional rate code $\theta_L$ equals to $log_{(1-P)}(1-P_{\det}^{1/M})$ and the number of blocks of the full rate code $\theta_H$ equals to $(B_F -  \theta_L B_L) / B_H$, the 2-layer rate-matched MSR code can guarantee that the same constraints for data regeneration (equation~(\ref{eq:error_correction}),~(\ref{eq:det_suc}) ) be satisfied for the data reconstruction.
\end{theorem}


\begin{proof} 
The maximum number of malicious nodes can be detected for the data reconstruction is no smaller than $M$: if $x>0.5$, the number is $\left \lfloor (n - \alpha - 1) /2 \right \rfloor$. We have $\left \lfloor (n - \alpha - 1) /2 \right \rfloor \ge \left \lfloor (n - xd - 1) /2 \right \rfloor = M$. If $x \le 0.5$, the number is $\left \lfloor (n -xd) /2 \right \rfloor$. We have $\left \lfloor (n -xd) /2 \right \rfloor \ge \left \lfloor (n - xd - 1) /2 \right \rfloor = M$.

The successful malicious node detection probability for the data reconstruction is larger than $P_{\det}$: the probability is $(1-(1-P)^{\alpha\theta_L})^M $, so we have $(1-(1-P)^{\alpha\theta_L})^M > (1-(1-P)^{\theta_L})^M \ge P_{\det}$.
\end{proof}

Although the rate-matching equation~(\ref{eq:rate-match}) does not apply to the data reconstruction, the reconstruction strategy in Algorithm~\ref{alg:rec_rate_matched} can still benefit from the different rates of the two codes. When $x \le 0.5$, the fractional rate code can detect and correct $\left \lfloor (n - xd) /2 \right \rfloor$ malicious nodes, which are more than  $\left \lfloor (n - d/2 - 1) /2 \right \rfloor$ malicious nodes that the full rate code can detect. When $x>0.5$, the full rate code and the fractional rate code can detect and correct the same number of malicious nodes:  $\left \lfloor (n - \alpha - 1) /2 \right \rfloor$.

From the analysis above we can see that the same optimized parameters, which are obtained for the data regeneration, can maintain the optimized trade-off between the malicious node detection and storage efficiency for the data reconstruction. 

\subsection{Performance Evaluation}
\label{Sec:performance}

From the analysis above, we know that for a distributed storage system with $n$ storage nodes out of which at most $M$ nodes are malicious, the 2-layer rate-matched MSR code can guarantee detection and correction of the malicious nodes during the data regeneration and reconstruction with the probability at least $P_{\det}$.

For a distributed storage system with $n=30$, $M=11$ and $P=0.2$, suppose we have a file with the size $B_F=14000M$ symbols to be stored in the system. The number of the fractional rate code blocks $\theta_L$ and the number of the full rate code blocks $\theta_H$ for different detection probabilities $P_{\det}$ are shown in Fig.~\ref{fig:number_of_blocks}. From the figure we can see that the number of fractional rate code blocks will increase when the detection probability becomes larger. Accordingly, the number of full rate code blocks will decrease.

\begin{figure}
\centering
\includegraphics[width=.8\columnwidth]{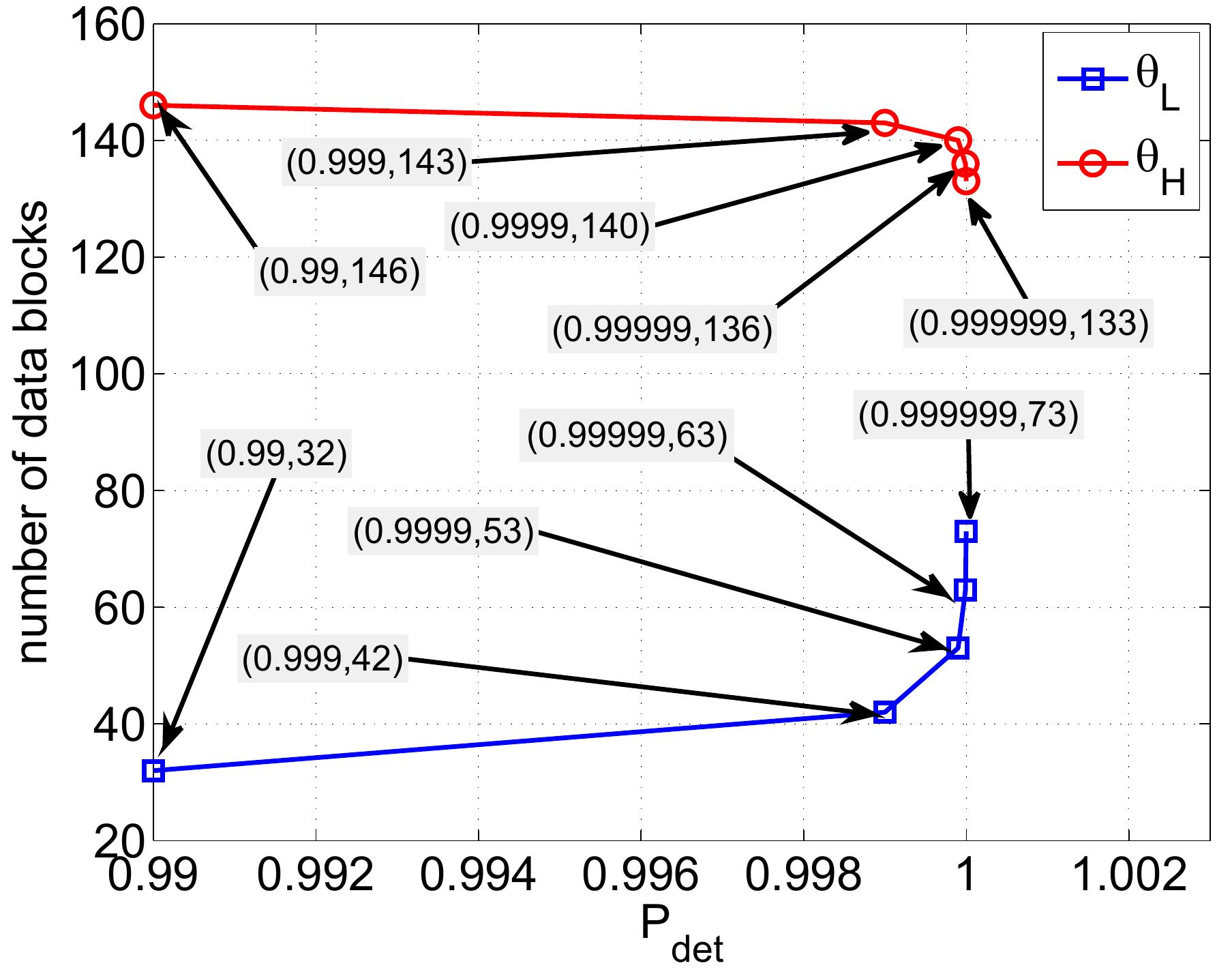}
\caption{The number of fractional/full rate code blocks for different $P_{\det}$}
\label{fig:number_of_blocks}
\end{figure}

For the universally resilient MSR code constructed in~\cite{Rashmi-err}, the efficiency of the code with the same regeneration performance as the 2-layer rate-matched MSR code is defined as 
\begin{equation}
{\delta'}\!\!_S = \frac{\alpha'(\alpha' + 1)}{\alpha' n} = \frac{\alpha'+1}{n} = \frac{xd/2+1}{n}.
\end{equation}
In Fig.~\ref{fig:efficiency_comp} we will show the efficiency ratios $\eta = \delta_S / {\delta'}\!\!_S$ between the 2-layer rate-matched MSR code and the universally resilient MSR code under different detection probabilities $P_{\det}$. From the figure we can see that the 2-layer rate-matched MSR code has higher efficiency than the universally resilient MSR code. When the successful malicious nodes detection probability is $0.999999$, the efficiency of the 2-layer rate-matched MSR code is about $70\%$ higher.

\begin{figure}
\centering
\includegraphics[width=.8\columnwidth]{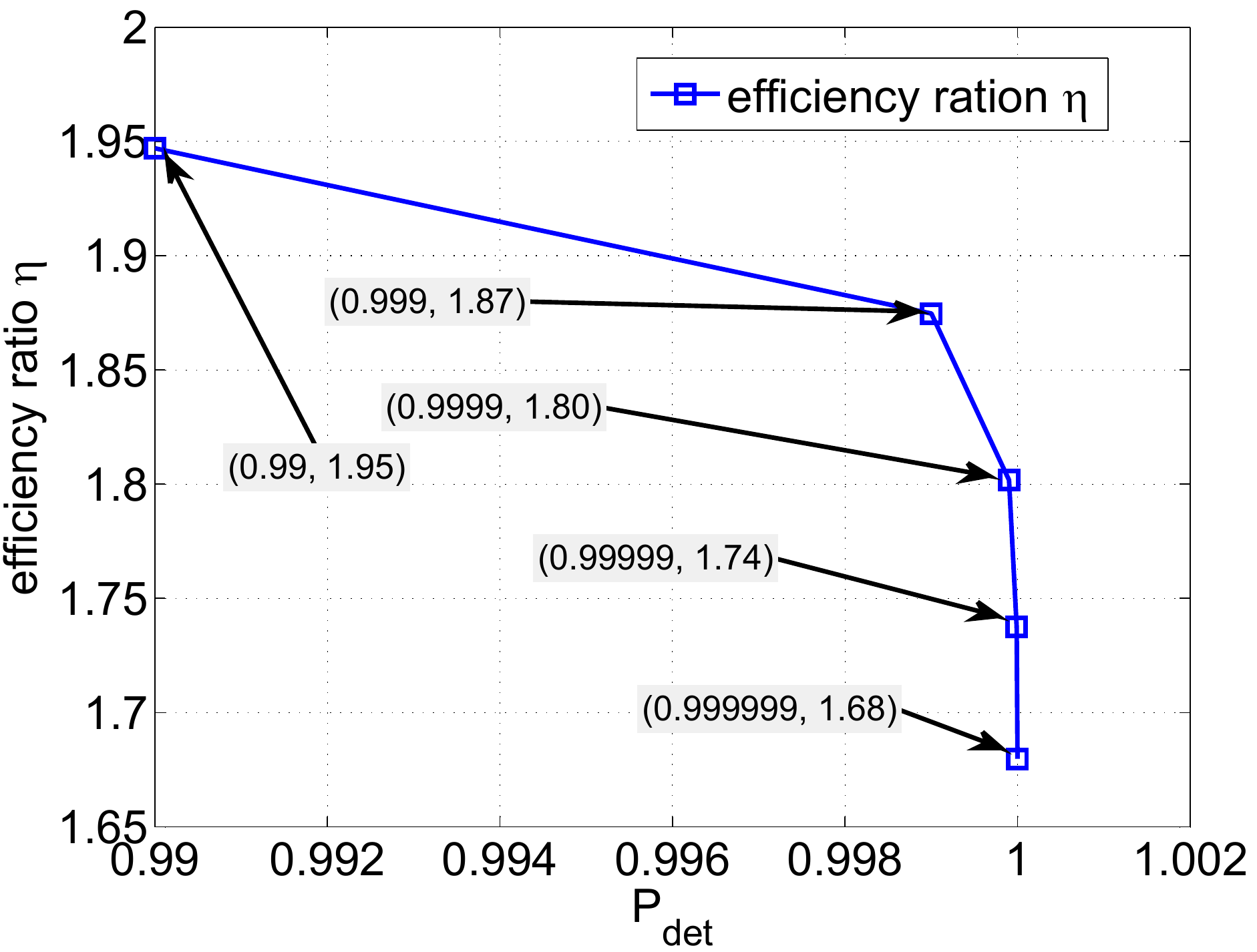}
\caption{Efficiency ratios between the 2-layer rate-matched MSR code and the normal error correction MSR code for different $P_{\det}$}
\label{fig:efficiency_comp}
\end{figure}

\section{$m$-Layer Rate-matched regenerating Code}\label{sec:m_layer_msr}

In this section, we will show our second optimization of the rate-matched MSR code: $m$-layer rate-matched MSR code. In the code design, we extend the design concept of the 2-layer rate-matched MSR code. Instead of encoding the data using two MSR codes with different match factors, we utilize $m$ layers of the full rate MSR codes with different parameter $d$, written as $d_i$ for layer $L_i$, $1 \leq i \leq m$, which satisfy 
\begin{equation}\label{eqn:m_layer_assump}
d_i \leq d_j,\:\: \forall 1 \leq i \leq j \leq m. 
\end{equation}
The data will be divided into $m$ parts and each part will be encoded by a distinct full rate MSR code. According to the analysis above, the code with a lower code rate has better error correction capability. 

The codewords will be decoded layer by layer in the order from layer $L_1$ to layer $L_m$. That is, the codewords encoded by the full rate MSR code with a lower $d$ will be decoded prior to those encoded by the full rate MSR code with a higher $d$. If errors were found by the full rate MSR code with a lower $d$, the corresponding nodes would be marked as malicious. The symbols sent from these nodes would be treated as erasures in the subsequent decoding of the full rate MSR codes with higher $d$'s. The purpose of this arrangement is to try to correct as many as erroneous symbols sent by malicious nodes and locate the corresponding malicious nodes using the full rate MSR code with a lower rate. However, the rates of the $m$ full rate MSR codes must match to achieve an optimal performance. Here we mainly focus on the rate-matching for data regeneration. We can see in the later analysis that the performance of data reconstruction can also be improved with this design criterion. 

The main idea of this optimization is to optimize the overall error correction capability by matching the code rates of different full rate MSR codes.

\subsection{Rate Matching and Parameters Optimization}\label{sec:rate_matching_m}
According to Section~\ref{Sec:Reg_high_rate}, the full rate MSR code $\CH_i$ for layer $L_i$ can be viewed as an $(n-1,d_i,n-d_i)$ MDS code for $1 \leq i \leq m$. During the optimization, we set the summation of the $d$'s of all the layers to a constant $d_0$:
\begin{equation}\label{eqn:sumd}
\sum_{i=1}^{m} d_i = d_0.
\end{equation}
Here we will show the optimization through an illustrative example first. Then we will present the general result.

\subsubsection{Optimization for $m=3$}
There are three layers of full rate MSR codes for $m=3$: $\CH_1$, $\CH_2$ and $\CH_3$. 

The first layer code $\CH_1$ can correct $t_1$ errors:
\begin{equation}\label{eqn:3_t1}
t_1 = \left \lfloor (n-d_1-1)/2 \right \rfloor = (n-d_1-1 - \varepsilon_1 )/2,
\end{equation}
where $\varepsilon_1 = 0$ or $1$ depending on whether $ (n-d_1-1)/2 $ is even or odd. 

By regarding the symbols from the $t_1$ nodes where errors are found by $\CH_1$ as erasures, the second layer code $\CH_2$ can correct $t_2$ errors:
\begin{equation}
\begin{array}{rcl}
t_2 & = & \left \lfloor (n-d_2 - 1 - t_1)/2 \right \rfloor + t_1  \\ 
 & = &  (n-d_2 - 1 - t_1 - \varepsilon_2)/2 + t_1  \\
& = & (2(n-d_2) + n - d_1 -2\varepsilon_2 - \varepsilon_1 - 3)/4, 
\end{array}
\end{equation}
where  $\varepsilon_2 = 0$ or $1$, with the restriction that $n-d_2-1 \ge t_1$, which can be written as:
\begin{equation}\label{eqn:l1l2}
-d_1 + 2d_2 \leq n + \varepsilon_1 - 1.
\end{equation}

The third layer code $\CH_3$ also treat the symbols from the $t_2$ nodes as erasures. $\CH_3$ can correct $t_3$ errors:
\begin{eqnarray}\label{eqn:opt_t3}
t_3 & = & \left \lfloor (n-d_3 - 1 - t_2)/2 \right \rfloor + t_2  \nonumber  \\ 
 & = &  (n-d_3 - 1 - t_2 - \varepsilon_2)/2 + t_2  \\
& = & (4(n-d_3) + 2(n-d_2) + n - d_1\! -\! 4\varepsilon_3\! -\! 2\varepsilon_2 \!-\! \varepsilon_1\! -\! 7)/8, \nonumber 
\end{eqnarray}
where  $\varepsilon_3 = 0$ or $1$, with the restriction that $n-d_3-1 \ge t_2$, which can be written as:
\begin{equation}\label{eqn:l2l3}
-d_1 - 2d_2 + 4d_3 \leq n + \varepsilon_1 + 2\varepsilon_2 - 1.
\end{equation}

According to the analysis above, the $d$'s of the three layers satisfy:
\begin{eqnarray}
d_1 - d_2 & \le & 0, \label{eqn:x1x2}  \\
d_2 - d_3 & \le & 0. \label{eqn:x2x3}  
\end{eqnarray}
And we can rewrite equation~(\ref{eqn:sumd}) as:
\begin{eqnarray}
d_1 + d_2 + d_3 & \le & d_0, \label{eqn:sumd_positive} \\
-d_1 - d_2 - d_3 & \le & -d_0. \label{eqn:sumd_negative} 
\end{eqnarray}

To maximize the error correction capability of the $m$-layer rate-matched MSR code for $m=3$, we have to maximize $t_3$, the number of errors that the third layer code $\CH_3$ can correct, since $t_3$ has included all the malicious nodes from which errors are found by the codes of first two layers. With all the constraints listed above, the optimization problem can written as:

\begin{equation}
\begin{array}{lll}
\mbox{Maximize} &\ & t_3\:\: \mbox{in}~(\ref{eqn:opt_t3}),  \\
\mbox{subject to} && (\ref{eqn:l1l2}),\: (\ref{eqn:l2l3}),\: (\ref{eqn:x1x2}),\: (\ref{eqn:x2x3}),\: (\ref{eqn:sumd_positive}),\: (\ref{eqn:sumd_negative}).
\end{array}
\end{equation}

Now we have changed this optimization problem into a typical linear programming problem. This linear programming problem has a feasible solution.  We solve it using the SIMPLEX algorithm~\cite{IA}. When $d_1 = d_2 = d_3 = \mbox{Round} (d_0/3)=\widetilde{d} $, the $m$-layer rate-matched MSR code can correct errors from at most
\begin{eqnarray}
\widetilde{t}_3 &=& (7n - 7\widetilde{d}  - 4\varepsilon_3 - 2\varepsilon_2 - \varepsilon_1- 7)/8 \nonumber \\
&\ge&  (7n - 7\widetilde{d}  -14)/8 \:\:\mbox{(worst case)}
\end{eqnarray}
malicious nodes, where Round($\cdot$) is the rounding operation.

\subsubsection{Evaluation of the Optimization for $m=3$}\label{sec:eva_3}
Similar to the storage efficiency $\delta_S$ defined in Section~\ref{sec:2-layer}, here we can define the error correction efficiency $\delta_C$ of the $m$-layer rate-matched MSR code as the ratio between the maximum number of malicious nodes that can be found and the total number of storage nodes in the network:
\begin{equation}
\delta_C = (7n - 7\widetilde{d} -14)/(8n).
\end{equation}
The universally resilient MSR code with the same code rate can be viewed as an $(n-1,\widetilde{d}, n - \widetilde{d})$ MDS code which can correct errors from at most $(n - \widetilde{d} - 1) /2$ malicious nodes (best case). So the error correction efficiency ${\delta'}\!\!_C$ is
\begin{equation}
{\delta'}\!\!_C =  ( n - \widetilde{d} - 1) / (2n).
\end{equation}
The comparison of the error correction capability between $m$-layer rate-matched MSR code for $m=3$ and universally resilient MSR code is shown in Fig.~\ref{fig:delta_c_3}. In this comparison, we set the number of storage nodes in the network $n=30$. From the figure we can see that the $m$-layer rate-matched MSR code for $m=3$ improves the error correction efficiency more than $50\%$.
\begin{figure}
\centering
\includegraphics[width=.8\columnwidth]{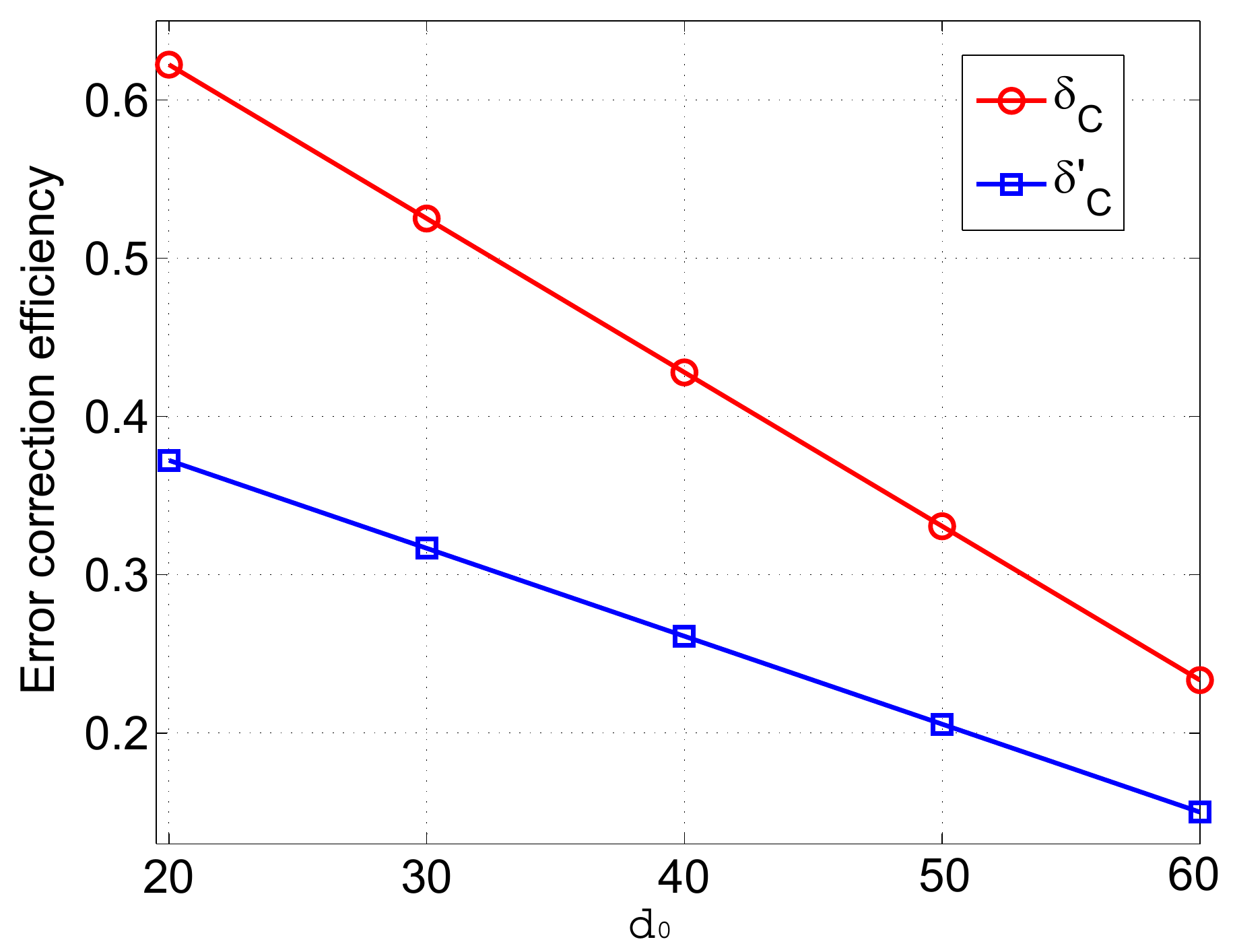}
\caption{Comparison of the error correction capability between $m$-layer rate-matched MSR code for $m=3$ and universally resilient MSR code}
\label{fig:delta_c_3}
\end{figure}

\subsubsection{General Optimization Result}
For the general $m$-layer rate-matched MSR code, the optimization process is similar. 

The first layer code $\CH_1$ can correct $t_1$ errors as in equation~(\ref{eqn:3_t1}). By regarding the symbols from the $t_{i-1}$ nodes where errors are found by $\CH_{i-1}$ as erasures, the $i^{th}$ layer code can correct $t_i$ errors for $2\le i \le m$:
\begin{eqnarray}
t_i & = & \left \lfloor (n-d_i - 1 - t_{i-1})/2 \right \rfloor + t_{i-1} \nonumber \\ 
 & = &  (n-d_i - 1 - t_{i-1} - \varepsilon_i)/2 + t_{i-1}  \label{eqn:opt_n_ti}\\
&=&\left(\sum_{j=1}^{i} 2^{j-1}(n-d_j) - \sum_{j=1}^i{2^{j-1}\varepsilon_j - 2^i+1}\right)/2^{i},\nonumber
\end{eqnarray}
where $\varepsilon_i=0$ or $1$, with the restriction that $n-d_i - 1 \ge t_{i-1} $, which can be written as:
\begin{equation}\label{eqn:restri_n}
-\sum_{j=1}^{i-1} 2^{j-1}d_j + 2^{i-1}d_i \leq n + \sum_{j=1}^{i-1} 2^{j-1}\varepsilon_j - 1.\\
\end{equation}

Similarly, the parameter $d$ of the $i^{th}$ layer for $2\le i\le m$ must satisfy
\begin{equation}\label{eqn:n_x1x2}
d_{i-1} - d_i \le 0.
\end{equation}
And equation~(\ref{eqn:sumd}) can be written as:
\begin{eqnarray}
\sum_{j=1}^{m}d_j & \le & d_0, \label{eqn:n_sumd_positive} \\
-\sum_{j=1}^{m}d_j & \le & -d_0. \label{eqn:n_sumd_negative} 
\end{eqnarray}

We can maximize the error correction capability of the $m$-layer rate-matched MSR code by maximizing $t_m$. With all the constrains listed above, the optimization problem can be written as:
\begin{equation}
\label{eqn:opt_original}
\begin{array}{lll}
\mbox{Maximize}  &\ & t_i\ \mbox{for}\ i=m \ \mbox{in}~(\ref{eqn:opt_n_ti}),  \\
\mbox{subject to} &\ & (\ref{eqn:restri_n}) \ \mbox{and} \ (\ref{eqn:n_x1x2}) \ \mbox{for} \ 2\le i \le m, \  (\ref{eqn:n_sumd_positive}),\ (\ref{eqn:n_sumd_negative}).
\end{array}
\end{equation}

After verifying that this linear programming problem has a feasible solution, we can use the SIMPLEX algorithm to solve it. The optimization result can be summarized as follows:
\begin{theorem} \label{thm:m-layer-reg}
For the regeneration of $m$-layer rate-matched MSR code, when 
\begin{equation}\label{eqn:opt_result_d}
d_i= \rm{Round}(d_0/m)=\widetilde{d}\:\:\:\: \mbox{for}\:\: 1 \le i \le m,
\end{equation}
it can correct errors from at most
\begin{eqnarray}
\widetilde{t}_m &=& ((2^{m}-1)(n-\widetilde{d}) -  \sum_{j=1}^m 2^{j-1}\varepsilon_j - 2^m + 1)/2^m\nonumber \\
&\ge&((2^m\!-\!1)(n\!-\!\widetilde{d})\! -\! 2^{m+1}\! +\! 2)/2^m \:\:\mbox{(worst case)}.
\end{eqnarray}
malicious nodes.
\end{theorem}

The error correction efficiency for the $m$-layer rate-matched MSR code is
\begin{equation}\label{eqn:efficiency_m_layer}
\delta_C = ((2^{m}-1)(n-\widetilde{d}) - 2^{m+1} + 2)/(2^{m}n).
\end{equation}
This is a monotonically increasing function for $m$, so we have:
\begin{corollary}\label{cor:monoincrease}
The error correction efficiency of the $m$-layer rate-matched MSR code increases with m, which is the number of layers.
\end{corollary}

\begin{remark}
During the optimization, we set the code rate of the rate-matched MSR code to a constant value and maximize the error correction capability. To optimizing the rate-matched MSR code, we can also set the error correction capability $t_i$ for $i=m$ in~(\ref{eqn:opt_n_ti}) to a constant value 
\begin{equation}
\label{eqn:dual_eqn}
t_m = t_0
\end{equation}
and maximize the code rate. The problem can be written as:
\begin{equation}
\begin{array}{l}
\mbox{Maximize} \:\:\:\:\:\sum_{j=1}^{m}d_j \\
\mbox{subject to} \:\:\:\:\: (\ref{eqn:restri_n})\:\: \mbox{and} \: (\ref{eqn:n_x1x2})\:\:\mbox{for}\:\:2\le i \le m, \:\: (\ref{eqn:dual_eqn}).
\end{array}
\end{equation}
The optimization result is the same as that of~(\ref{eqn:opt_original}): when all the $d_i's$ for $1 \leq i \leq m$ are the same, the code rate is maximized. $d_i$, $1 \le i \le m$, satisfies the following equation:
\begin{equation}
d_i \ge n - \frac{2^mt_0+2^{m+1}-2}{2^m-1} \:\:\mbox{(worst case)}.
\end{equation}
\end{remark}

\subsubsection{Evaluation of the Optimization}
Although at the beginning of this section we propose to decode the code with a lower rate first in the $m$-layer rate-matched MSR code, equation~(\ref{eqn:opt_result_d}) shows that we can get the optimized error correction capability when all the rates of the codes in the $m$-layer code are equal. However, this result is not in conflict with our assumption in equation~(\ref{eqn:m_layer_assump}).

\paragraph{Comparison with the Hermitian code based MSR code in~\cite{jian14}}
The Hermitian code based MSR code (H-MSR code) in~\cite{jian14} has better error correction capability than the universally resilient MSR code. However, because the structure of the underlying Hermitian code is predetermined, the error correction capability might not be optimal. In figure~\ref{fig:comp_h}, the maximum number of malicious nodes from which the errors can be corrected by the H-MSR code is shown. Here we set the parameter $q$ of the Hermitian code~\cite{Hermitian} from 4 to 16 with a step of 2. In the figure, we also plot the performance of the $m$-layer rate-matched MSR code with the same code rates as the H-MSR code. The comparison result demonstrates that the rate-matched MSR code has better error correction capability than the H-MSR code. Moreover, the rate-matched code is easier to understand and has more flexibility than the H-MSR code.

\begin{figure}
\centering
\includegraphics[width=.8\columnwidth]{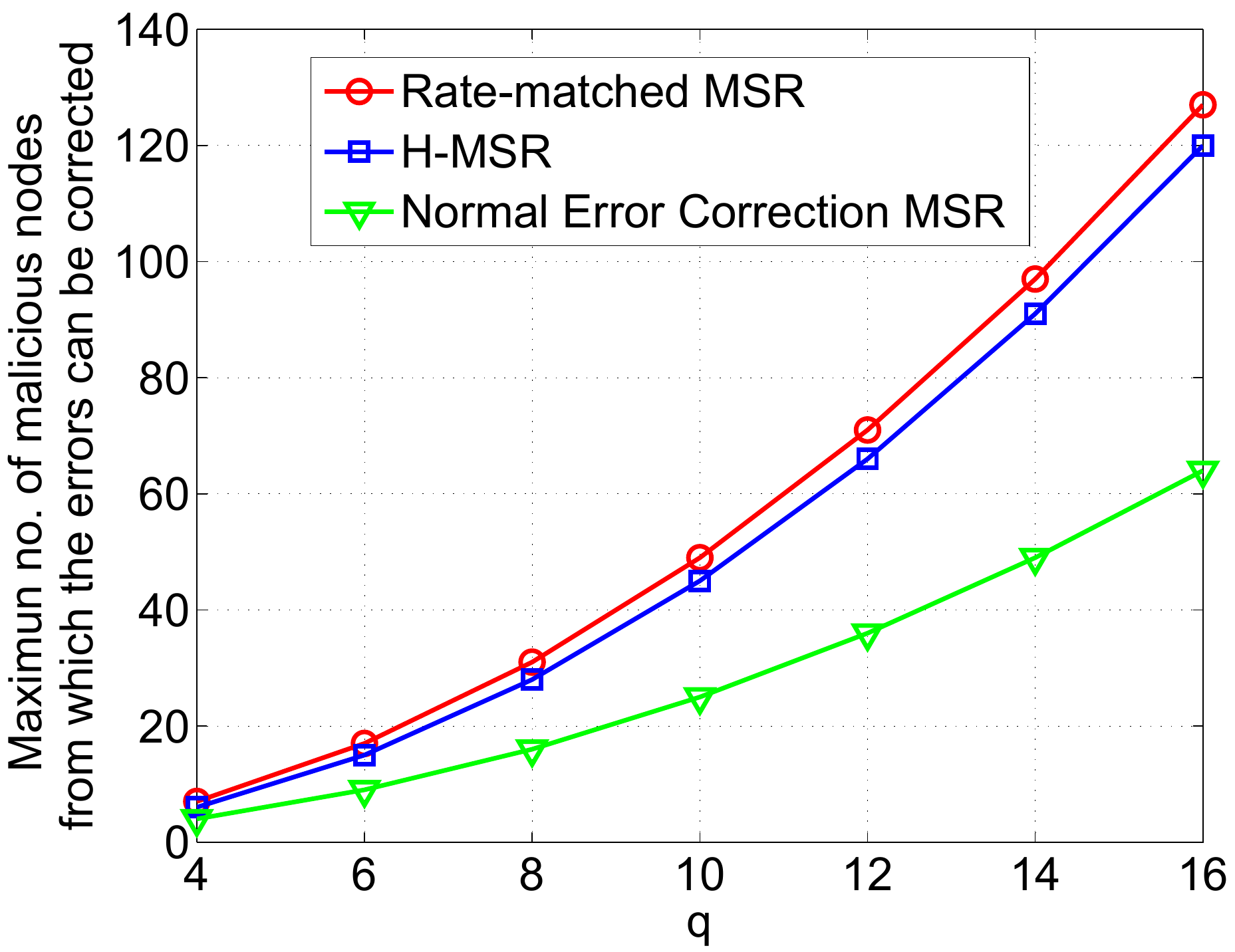}
\caption{Comparison of error correction capability between the $m$-layer rate matched MSR code and the H-MSR code}
\label{fig:comp_h}
\end{figure}

\paragraph{Number of layers and error correction efficiency}

Since we have seen the advantage of the rate-matched MSR code over the universally resilient MSR code in Section~\ref{sec:eva_3}, here we will mainly discuss how the number of layers can affect the error correction efficiency. The error correction efficiency of the $m$-layer rate-matched MSR code is shown is Fig.~\ref{fig:delta_c_n}, where we set $n=30$ and $d_0=50$. We also plot the error correction efficiency ${\delta'}_C$ of the universally resilient MSR code with same code rates for comparison. From the figure we can see that when $n$ and $d_0$ are fixed, the optimal error correction efficiency will increase with the number of layers $m$ as in Corollary~\ref{cor:monoincrease}.
\begin{figure}
\centering
\includegraphics[width=.8\columnwidth]{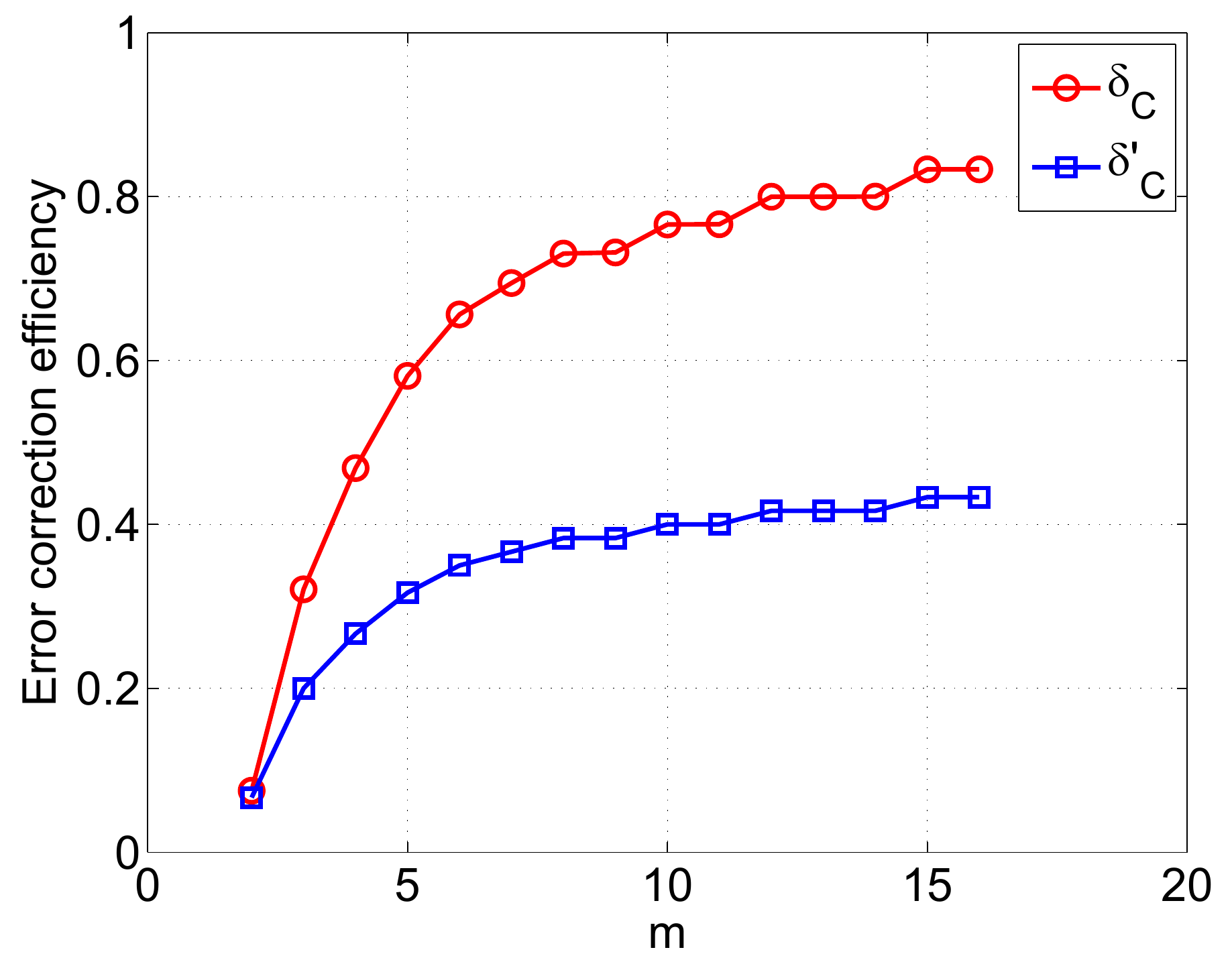}
\caption{The optimal error correction efficiency of the $m$-layer rate-matched MSR code under different m for $2 \le m \le 16$}
\label{fig:delta_c_n}
\end{figure}

\paragraph{Optimized storage capacity}
Moreover, the optimization condition in equation~(\ref{eqn:opt_result_d}) also leads to maximum storage capacity besides the optimal error correction capability. We have the following theorem:
\begin{theorem}
The $m$-layer rate-matched MSR code can achieve the maximum storage capacity if the parameter $d$'s of all the layers are the same, under the constraint in equation~(\ref{eqn:sumd}).
\end{theorem}
\begin{proof}
The code of the $i^{th}$ layer can store one block of data with the size $B_i = \alpha_i(\alpha_i +1) =(d_i/2)(d_i/2 + 1)  $. So the $m$-layer code can store data with the size $B=\sum_{i=1}^{m}(d_i/2)(d_i/2 + 1) $. Our goal here is to maximize $B$ under the constraint in equation~(\ref{eqn:sumd}). 

We can use Lagrange multipliers to find the point of maximum $B$. Let
\begin{equation}\label{eqn:largrange_m}
\Lambda_L(d_1,\dots,d_m,\lambda) = \sum_{i=1}^{m}(d_i/2)(d_i/2 + 1) + \lambda(\sum_{i=1}^{m}d_i - d_0).
\end{equation}
We can find the maximum value of $B$ by setting the partial derivatives of this equation to zero:
\begin{equation}
\frac{\partial \Lambda_L}{\partial d_i} = \frac{d_i+1}{2} - \lambda = 0, \:\: \forall 1\le i \le m.
\end{equation}
Here we can see that when all the parameter $d$'s of all the layers are the same, we can get the maximum storage capacity $B$. This maximization condition coincides with the optimization condition for achieving the goal of this section: optimizing the overall error correction capability of the rate-matched MSR code.
\end{proof}

\subsection{Practical Consideration of the Optimization}
So far, we implicitly presume that there is only one data block of the size $B_i = \alpha_i(\alpha_i +1)$ for each layer $i$. In practical distributed storage, it is the parameter $d_i$ that is fixed instead of $d_0$, the summation of $d_i$. However, as long as we use $m$ layers of MSR codes with the same parameter $d=\widetilde{d}$, we will still get the optimal solution for $d_0 = m\widetilde{d}$. In fact, the $m$-layer rate-matched MSR code here becomes a single full rate MSR code with parameter $d=\widetilde{d}$ and $m$ data blocks. And based on the dependent decoding idea we describe at the beginning of Section~\ref{sec:m_layer_msr}, we can achieve the optimal performance.

So when the file size $B_F$ is larger than one data block size $\widetilde{B}$ of the single full rate MSR code with parameter $d=\widetilde{d}$, we will divide the file into $\left. \lceil B_F/\widetilde{B} \rceil  \right.$ data blocks and encode them separately. If we decode these data blocks dependently, we can get the optimal error correction efficiency.

\subsubsection{Evaluation of the Optimal Error Correction Efficiency}
In the practical case, $\widetilde{d}$ in equation~(\ref{eqn:efficiency_m_layer}) is fixed. So here we will study the relationship between the number of dependently decoding data blocks $m$ and the error correction efficiency $\delta_C$, which is shown in Fig.~\ref{fig:delta_c_practical}. We set $n=30$ and $\widetilde{d}=5,10$. From the figure we can see that although $\delta_C$ will become higher with the increasing of dependently decoding data blocks $m$, the efficiency improvement will be negligible for $m\ge8$.  Actually when $m=7$ the efficiency has already become $99\%$ of the upper bound of $\delta_C$.

\begin{figure}
\centering
\includegraphics[width=.8\columnwidth]{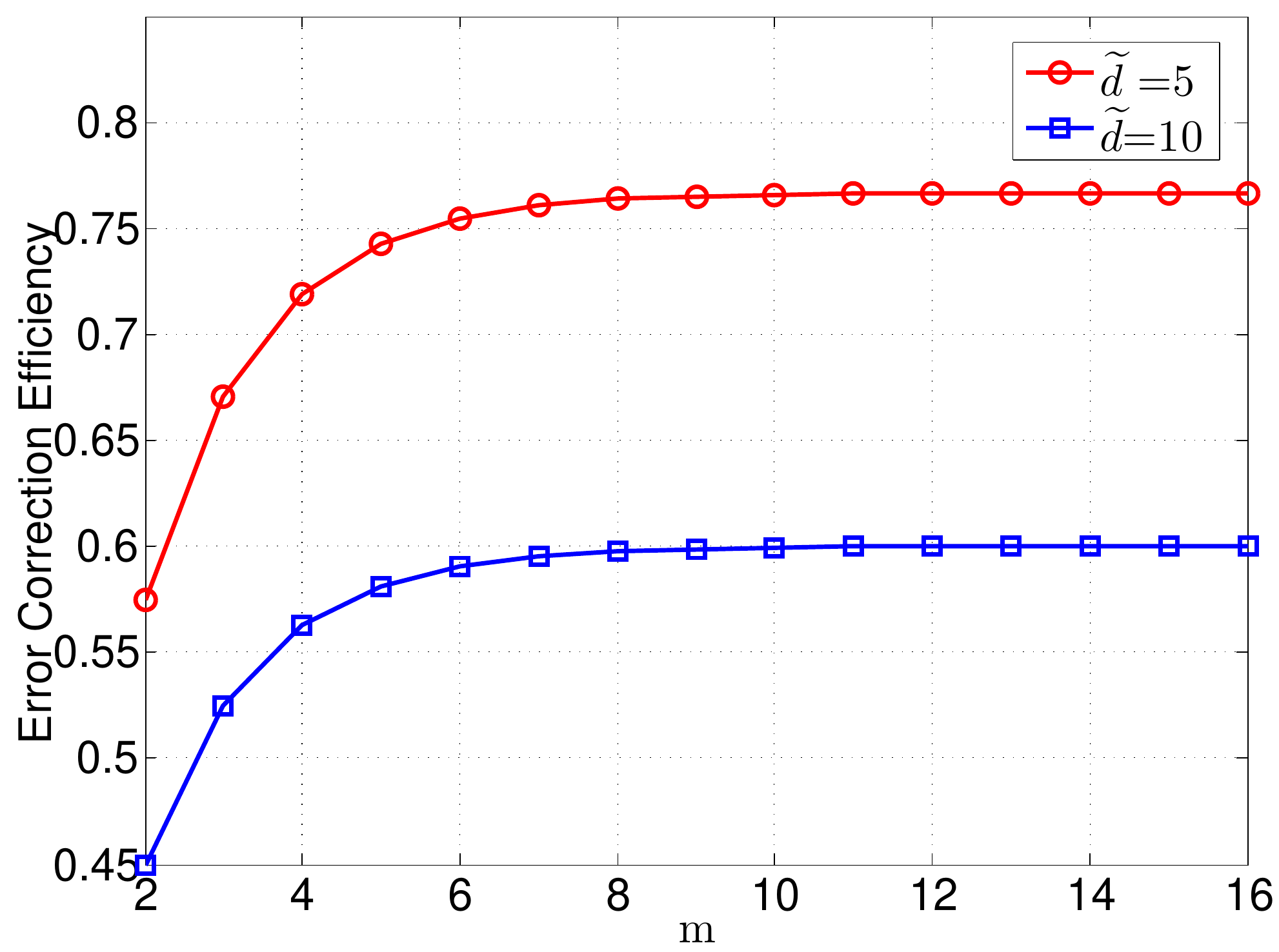}
\caption{The optimal error correction efficiency for $2 \le m \le 16$}
\label{fig:delta_c_practical}
\end{figure}

On the other hand, there exist parallel algorithms for fast MDS code decoding~\cite{Dabiri95}. We can decode blocks of MDS codewords parallel in a pipeline fashion to accelerate the overall decoding speed. The more blocks of codewords we decode parallel, the faster we will finish the whole decoding process. For large files that could be divided into a large amount of data blocks ($\theta$ blocks), we can get a trade-off between the optimal error correction efficiency and the decoding speed by setting the number of dependently decoding data blocks $m$ and the number of parallel decoding data blocks $\rho$ under the constraint $\theta = m\rho$.

\subsection{Encoding}
From the analysis above we know that to encode a file with size $B_F$ using the optimal $m$-layer rate-matched MSR code is to encode the file using a full rate MSR code with predetermined parameter $d=2\alpha=\widetilde{d}$. First the file will be divided into $\theta$ blocks of data with size $\widetilde{B}$, where $\theta= \left. \lceil B_F / \widetilde{B} \rceil \right.$. Then the $\theta$ blocks of data will be encoded into code matrices $\CH_1,\dots,\CH_{\theta}$ and form the final $n \times \alpha\theta$ codeword matrix: $CM = [\CH_1,\dots,\CH_{\theta}]$. Each row $\mathbf{c}_i = [\ch_{1,i}, \dots, \ch_{\theta,i}]$, $0 \leq i \leq n-1$, of the codeword matrix $CM$ will be stored in storage node $i$, where $\ch_{j,i}$ is the $i^{th}$ row of $\CH_j$, $1 \leq j \leq \theta$. The encoding vector $\boldsymbol{\psi}_i$ for storage node $i$ is the $i^{th}$ row of $\Psi$ in equation~(\ref{eq:encoding_msr_h}).
\begin{theorem}
The encoding of m-layer rate-matched MSR code can achieve the MSR point in equation~(\ref{eq:MSR_tradeoff}) since both the full rate code and the fractional code are MSR codes.
\end{theorem}

\subsection{Regeneration}
Suppose node $z$ fails, the replacement node $z'$ will send regeneration requests to the rest of $n-1$ helper nodes. Upon receiving the regeneration request, helper node $i$ will calculate and send out the help symbol $p_i = {\bf{c}}_i \boldsymbol{\phi}_z^T$. 

\begin{figure}
\centering
\includegraphics[width=.8\columnwidth]{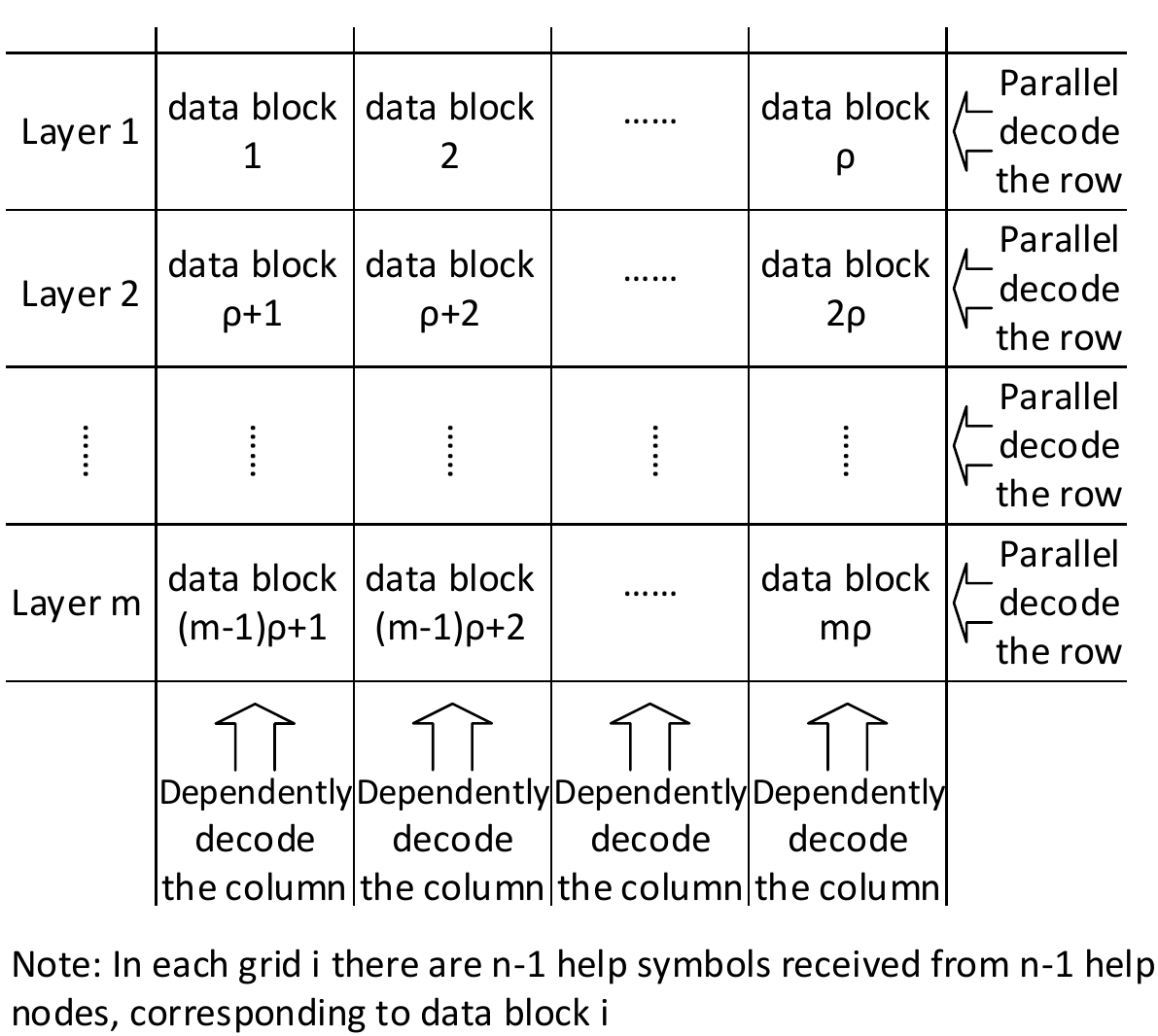}
\caption{Lattice of received help symbols for regeneration}
\label{fig:lattice_arrange}
\end{figure}

As we discuss above, combining both dependent decoding and parallel decoding can achieve the trade-off between optimal error correction efficiency and decoding speed. Although all $\theta$ blocks of data are encoded with the same MSR code, $z'$ will place the received help symbols into a 2-dimension lattice with size $m \times \rho$ as shown in Fig.~\ref{fig:lattice_arrange}. In each grid of the lattice there are $n-1$ help symbols corresponding to one data block, received from $n-1$ helper nodes. We can view each row of the lattice as related to a layer of an $m$-layer rate-matched MSR code with $\rho$ blocks of data, which will be decoded parallel. We also view each column of the lattice as related to $m$ layers of an $m$-layer rate-matched MSR code with one block of data each layer, which will be decoded dependently. $z'$ will perform Algorithm~\ref{alg:reg_m_layer} to regenerate the contents of the failed node $z$. 

Arrange the received help symbols according to Fig.~\ref{fig:lattice_arrange}. Repeat the following steps from Layer $1$ to Layer $m$:
\begin{algorithms}$z'$ regenerates symbols of the failed node $z$ for the $m$-layer rate-matched MSR code
\label{alg:reg_m_layer}

\normalfont

\begin{namelist}{\textbf{Step n:}}
\item [\step{1}] For a certain grid, if errors are detected in the symbols sent by node $i$ in previous layers of the same column, replace the symbol sent from node $i$ by an erasure $\bigotimes$.

\item [\step{2}] Parallel regenerate all the symbols related to $\rho$ data blocks using the algorithm similar to Algorithm~\ref{alg:reg_h} with only one difference: parallel decode all the $\rho$ MDS codes in {\bf{Step 1}} of Algorithm~\ref{alg:reg_h}. 
\end{namelist}
\end{algorithms}

The error correction capability of the regeneration is described in Theorem~\ref{thm:m-layer-reg}.

\subsection{Reconstruction}
When DC needs to reconstruct the original file, it will send reconstruction requests to $n$ storage nodes. Upon receiving the request, node $i$ will send out the symbol vector $\mathbf{c}_i$. Suppose $\mathbf{c}'_i= \mathbf{c}_i + \mathbf{e}_i$ is the response from the $i^{th}$ storage node. If $\mathbf{c}_i$ has been modified by the malicious node $i$, we have $\mathbf{e}_i \in \mathbb{F}_q^{\alpha\theta} \backslash \{\mathbf{0}\}$. The strategy of combining dependent decoding and parallel decoding for reconstruction is similar to that for regeneration. $DC$ will place  the received symbols into a 2-dimension lattice with size $m \times \rho$. The only difference is that in a grid of the lattice there are $n$ symbol vectors $\ch'_{j,0}, \dots, \ch'_{j,n-1}$ corresponding to data block $j$, received from $n$ storage nodes. DC will perform the reconstruction using Algorithm~\ref{alg:rec_rate_matched_m_rec}.

Arrange the received symbols similar to Fig.~\ref{fig:lattice_arrange}. Here we place received codeword matrix $\CH'_j$ into grid $j$ instead of help symbols received from n-1 help nodes. Repeat the following steps from Layer $1$ to Layer $m$:

\begin{algorithms}DC reconstructs the original file for the $m$-layer rate-matched MSR code
\label{alg:rec_rate_matched_m_rec}

\normalfont

\begin{namelist}{\textbf{Step n:}}
\item [\step{1}] For a certain grid, if errors are detected in the symbols sent by node $i$ in previous layers of the same column, replace symbols sent from node $i$ by erasures $\bigotimes$.

\item [\step{2}] Parallel reconstruct all the symbols of the $\rho$ data blocks using the algorithm similar to Section~\ref{Sec:Rec_high_rate} with only one difference: parallel decode all the MDS codes in Section~\ref{Sec:Rec_high_rate}. 
\end{namelist}
\end{algorithms}


For data reconstruction, we have the following theorem:
\begin{theorem}[Optimized Parameters]
For the reconstruction of $m$-layer rate-matched MSR code, when 
\begin{equation}\label{eqn:opt_result_d}
d_i= \rm{Round}(d_0/m)=\widetilde{d}\:\:\:\: \mbox{for}\:\: 1 \le i \le m,
\end{equation}
the number of malicious nodes from which the errors can be corrected is maximized.
\end{theorem}

\begin{proof}
From Section~\ref{sec:rate_matching_m} we know that for regeneration of an optimal $m$-layer rate-matched MSR code, the parameter $d$'s of all the layers are the same, which implies the parameter $\alpha$'s of all layers are also the same. Since the optimization of regeneration is derived based on the decoding of $(n-1,d,n-d)$ MDS codes and in reconstruction we have to decode $(n-1,\alpha,n-\alpha)$ MDS codes, if the parameter $\alpha$'s of all the layers are the same, we can achieve the same optimization results for reconstruction.
\end{proof}


\section{Conclusion} \label{Sec:Conclusion}
In this paper, we develop two rate-matched regenerating codes for malicious nodes detection and correction in hostile networks: 2-layer rate-matched regenerating code and $m$-layer rate-matched regenerating code. We propose the encoding, regeneration and reconstruction algorithms for both codes. For the 2-layer rate-matched code, we optimize the parameters for the data regeneration, considering the trade-off between the malicious nodes detection probability and the storage efficiency. Theoretical analysis shows that the code can successfully detect and correct malicious nodes using the optimized parameters. Our analysis also shows that the code has higher storage efficiency compared to the universally resilient regenerating code ($70\%$ higher for the detection probability $0.999999$). Then we extend the 2-layer code to $m$-layer code and optimize the overall error correction efficiency by matching the code rate of each layer's regenerating code. Theoretical analysis shows that the optimized parameter could also achieve the maximum storage capacity under the same constraint. Furthermore, analysis shows that compared to the universally resilient regenerating code, our code can improve the error correction efficiency more than $50\%$.

\end{document}